\documentclass[fleqn,preprint]{elsarticle}

\usepackage{amsmath,amssymb}

\usepackage{tikz}
\usetikzlibrary{calc,positioning,arrows}
\usepackage{graphicx}
\usepackage{float}

\newtheorem{theorem}{Theorem}
\newtheorem{lemma}{Lemma}
\newtheorem{claim}{Claim}

\newtheorem{fact}{Fact}
\newenvironment{proof}{\noindent {\bf Proof.\,}}{\qed}

\def\MC{\textsc{Matching Cut}}
\def\d{\mathrm{dist}}

\newcommand{\Po}{\mathsf{P}}
\newcommand{\NP}{\mathsf{NP}}

\begin{document}

\begin{frontmatter}

\title{A complexity dichotomy for Matching Cut in (bipartite) graphs of fixed diameter\tnoteref{label1}}
\tnotetext[label1]{This paper is an extended version of the ISAAC 2016 paper~\cite{LeL16}.}

\author{Hoang-Oanh Le}
\ead{LeHoangOanh@web.de}
\address{Berlin, Germany}

\author{Van Bang Le}
\ead{van-bang.le@uni-rostock.de}
\address{Universit\"at Rostock, Institut f\"ur Informatik, Rostock, Germany}

\begin{abstract}
In a graph, a matching cut is an edge cut that is a matching. \MC\ is the problem of deciding whether or not a given graph has a matching cut, which is known to be NP-complete even when restricted to bipartite graphs. It has been proved that \MC\ is polynomially solvable for graphs of diameter two. 
In this paper, we show that, for any fixed integer $d\ge 3$, \MC\ is $\NP$-complete in the class of graphs of diameter $d$. This resolves an open problem posed by Borowiecki and Jesse-J\'ozefczyk (2008)~\cite{BorowieckiJ08}. 

We then show that, for any fixed integer $d\ge 4$, \MC\ is $\NP$-complete even when restricted to the class of bipartite graphs of diameter $d$. Complementing the hardness results, we show that \MC\ is polynomial-time solvable in the class of bipartite graphs of diameter at most three, and point out a new and simple polynomial-time algorithm solving \MC\ in graphs of diameter~$2$. 
\end{abstract}


\end{frontmatter}

\section{Introduction}

In a graph $G=(V,E)$, a \emph{cut} is a partition $V=X\,\dot\cup\,Y$ of the vertex set into disjoint, non-empty sets $X$ and $Y$, written $(X,Y)$. The set of all edges in $G$ having an endvertex in $X$ and the other endvertex in $Y$, also written $(X,Y)$, is called the \emph{edge cut} of the cut $(X,Y)$. A \emph{matching cut} is an edge cut that is a (possibly empty) matching. Note that, by our definition, a matching whose removal disconnects the graph need not be a matching cut (but such a matching always contains some matching cut). Note also that a graph has an empty matching cut if and only if it is disconnected. 

Another way to define matching cuts is as follows (\cite{Graham,Chvatal84}). A partition $V=X\,\dot\cup\,Y$ of the vertex set of the graph $G=(V,E)$ into disjoint, non-empty sets $X$ and $Y$, is a matching cut if and only if each vertex in $X$ has at most one neighbor in $Y$ and each vertex in $Y$ has at most one neighbor in $X$.

Graham~\cite{Graham} studied matching cuts in graphs in connection to a number theory problem called cube-numbering. In~\cite{FarleyP82}, Farley and Proskurowski studied matching cuts in the context of network applications. Patrignani and Pizzonia~\cite{PatrignaniP01} pointed out an application of matching cuts in graph drawing. Matching cuts have been used by Ara\'ujo et al.~\cite{ACGH} in studying good edge-labellings in the context of WDM (Wavelength Division Multiplexing) networks.

Not every graph has a matching cut; the \MC\ problem is the problem of deciding whether or not a given graph has a matching cut:

\medskip\noindent
\fbox{
\begin{minipage}{.955\textwidth}
\MC\\[.5ex]
\begin{tabular}{l l}
{\em Instance:}& A graph $G=(V,E)$.\\
{\em Question:}& Does $G$ have a matching cut?\\
\end{tabular}
\end{minipage}
}

\medskip\noindent
This paper considers the computational complexity of the \MC\ problem in graphs of fixed diameter.  

\subsection{Previous results}  
Graphs admitting a matching cut were first discussed by Graham in \cite{Graham} under the name \emph{decomposable graphs}. The first complexity and algorithmic results for \MC\ have been obtained by Chv\'atal, who proved in \cite{Chvatal84} that \MC\ is $\NP$-complete, even when restricted to graphs of maximum degree four and polynomially solvable for graphs of maximum degree at most three. These results triggered a lot of research on the computational complexity of \MC\ in graphs with additional structural assumptions; see~\cite{Bonsma09,BorowieckiJ08,KratschL16,LeR03,Moshi89,PatrignaniP01}. In particular, the $\NP$-hardness of \MC\ has been further strengthened for planar graphs of maximum degree four (\cite{Bonsma09}) and bipartite graphs of maximum degree four (\cite{LeR03}). Moreover, it follows from Bonsma's result~\cite{Bonsma09} and a simple reduction observed by Moshi~\cite{Moshi89} that \MC\ remains NP-complete on bipartite planar graphs of maximum degree eight.  

On the positive side, among others, an important polynomially solvable case has been established by Borowiecki and Jesse-J\'ozefczyk, who proved in~\cite{BorowieckiJ08} in 2008 that \MC\ is polynomial-time solvable for graphs of diameter $2$. They also posed the problem of determining the largest integer $d$ such that \MC\ is solvable in polynomial time for graphs of diameter $d$. This open problem was the main motivation of the present paper. 

\subsection{Our contributions} We prove that \MC\ is $\NP$-complete, even when restricted to graphs of diameter $d$, for any fixed integer $d\ge 3$. Thus, unless $\NP =\Po$, \MC\ cannot be solved in polynomial time for graphs of diameter $d$, for any fixed $d\ge 3$. This resolves the open problem posed by Borowiecki and Jesse-J\'ozefczyk mentioned above. Actually, we show a little more: \MC\ is $\NP$-complete in graphs of diameter $3$ and remains $\NP$-complete in bipartite graphs of fixed diameter $d\ge 4$. An alternative proof (reduction from 1-IN-3 3SAT) for the case of graphs with diameter $d\ge 4$ and the case of bipartite graphs of diameter $d\ge 5$ is given in the conference paper~\cite{LeL16}. 
Complementing our hardness results, we show that \MC\ can be solved in polynomial time in bipartite graphs of diameter at most $3$. We also point out a new and simple approach solving \MC\ in diameter-2 graphs in polynomial time. In summary, our main results are the following complexity dichotomy theorems: 
\begin{itemize}
\item \MC\ is $\NP$-complete for graphs of fixed diameter $d\ge 3$ and (we provide an alternative proof that the problem is) polynomially solvable for graphs of diameter $d\le 2$.
\item \MC\ is $\NP$-complete for bipartite graphs of fixed diameter $d\ge 4$ and polynomially solvable for bipartite graphs of diameter $d\le 3$. 
\end{itemize}

\subsection{Notation and terminology}
Let $G=(V,E)$ be a graph with vertex set $V(G)=V$ and edge set $E(G)=E$. 
An \emph{independent set} (a \emph{clique}) in $G$ is a set of pairwise
non-adjacent (adjacent) vertices. A \emph{biclique} is a complete bipartite graph; we sometimes write $Q= (V_1,V_2)$ for a biclique with color classes $V_1$ and $V_2$.  
The neighborhood of a vertex $v$ in $G$, denoted by $N_G(v)$, is the set of all vertices in $G$ adjacent to $v$; if the context is clear, we simply write $N(v)$. 
Set $\deg(v) = |N(v)|$, the degree of the vertex $v$. 
For a subset $W\subseteq V$,
$G[W]$ is the subgraph of $G$ induced by $W$, and $G-W$ stands for $G[V\setminus W]$.  
The complete graph and the path on $n$ vertices is denoted by $K_n$ and $P_n$, respectively; $K_3$ is also called a \emph{triangle}. The complete bipartite graph with one color class of size $p$ and the other of size $q$ is denoted by $K_{p,q}$. Observe that, for any matching cut $(X,Y)$ of $G$, any $K_n$ with $n\ge 3$, and any $K_{p,q}$ with $p\ge 2, q\ge 3$, in $G$ is contained in $G[X]$ or else in $G[Y]$. 

Given a graph $G=(V,E)$ and a partition $V=X\,\dot\cup\, Y$, it can be decided in linear time if $(X,Y)$ is a matching cut of $G$. This is because $(X,Y)$ is a matching cut of $G$ if and only if the bipartite subgraph $B_G(X,Y)$ of $G$ with the color classes $X$ and $Y$ and edge set $(X,Y)$ is $P_3$-free. That is, $(X,Y)$ is a matching cut of $G$ if and only if the non-trivial connected components of the bipartite graph $B_G(X,Y)$ are edges. A path $P_3$ in $B_G(X,Y)$, if any, is called a \emph{bad~$P_3$}.

A \emph{bridge} in a graph is an edge whose deletion increases the number of the connected components. Since disconnected graphs and graphs having a bridge have a matching cut, we may assume that all graphs considered are connected and $2$-edge connected. The \emph{distance} between two vertices $u,v$ in a (connected) graph $G$, denoted $\d_G(u,v)$, is the length of a shortest path connecting $u$ and $v$. 
The \emph{diameter} of $G$, denoted $\mathrm{diam}(G)$, is the maximum distance between all pairs of vertices in $G$.

\smallskip
The paper is organized as follows. In Section~\ref{sec:diam3} we show that \MC\ is $\NP$-complete when restricted to graphs of diameter $d\ge 3$, for any fixed integer $d\ge 3$. In Section~\ref{sec:diam4bip} we show that \MC\ remains $\NP$-complete even when restricted to bipartite graphs of diameter $d$, for any fixed integer $d\ge 4$. In section~\ref{sec:diam3bip} we point out a new and simple polynomial time algorithm solving \MC\ in diameter-$2$ graphs, and show that \MC\ can be solved in polynomial time for bipartite graphs of diameter at most $3$. We conclude the paper with Section~\ref{sec:con}.


\section{Matching Cut in graphs of fixed diameter $d\ge 3$}\label{sec:diam3}


In this section, we first reduce \MC\ to \MC\ restricted to graphs of diameter~$3$. 

Given an instance $G=(V,E)$ of \MC, construct a new graph $G'=(V',E')$ as follows. 
First, let $V=\{v_1,v_2,\ldots, v_n\}$, and for each $1\le i\le n$, let $Q_i$ be a complete graph on vertex set $\{q^1_i, \ldots, q^n_i\}$. Then, $G'$ is obtained from $G$ and $Q_i, 1\le i\le n$, by adding edges between $v_i$ and all vertices in $Q_i$ (thus, for each $i$, $V(Q_i)\cup\{v_i\}$ induces a clique in $G'$) and edges between the vertex $q^j_i\in Q_i$ and the vertex $q^i_j\in Q_j$ for any pair $i,j\in\{1,\ldots,n\}, i\not=j$. Formally,

\begin{align*}
  V' &=V \cup \bigcup_{1\le i\le n} V(Q_i),\\ \smallskip
  E' &=E \cup \bigcup_{1\le i\le n} \{v_iq \mid q\in V(Q_i)\} \cup \bigcup_{1\le i\le n} E(Q_i) \cup \{q^j_iq^i_j\mid 1\le i,j\le n, i\not=j\}.
\end{align*}

See also Figure~\ref{fig:ij}. Clearly, $G'$ can be constructed from $G$ in $O(n^3)$ steps. 

\begin{figure}[htb]
\begin{center}
\begin{tikzpicture}[scale=.39]
\tikzstyle{vertexG}=[circle,inner sep=1.5pt,fill=black];
\tikzstyle{vertexY}=[draw,circle,inner sep=1.5pt];
\tikzstyle{vertex}=[draw,circle,inner sep=1.5pt]; 

\filldraw[fill=black!5!white, draw=gray] (4.5,9.4) rectangle (21.5,12.5);

\node[vertexG] (vi) at (6,10)  [label=left:$v_i$] {}; 
\node at (0.7,6) {$Q_i$};
\node[vertex] (qi1) at (2,6)  {}; 
\node[vertex] (qi2) at (4,6)  {}; 
\node[vertex] (qij) at (6,6)  {}; 
\node[] (q) at (5.4,6.7) {\small $q^j_i$}; 
\node[vertex] (qi4) at (8,6)  {}; 
\node[vertex] (qi5) at (10,6)  {}; 

\node[vertexG] (vj) at (20,10)  [label=right:$v_j$] {}; 
\node at (25.3,6) {$Q_j$};
\node[vertex] (qj1) at (16,6)  {}; 
\node[vertex] (qj2) at (18,6)  {}; 
\node[vertex] (qji) at (20,6)  {}; 
\node[] (q') at (20.6,6.7) {\small $q^i_j$}; 
\node[vertex] (qj4) at (22,6)  {}; 
\node[vertex] (qj5) at (24,6)  {}; 

\foreach \x in {qi1,qi2,qij,qi4,qi5}
{
\draw[gray, thin] (vi) -- (\x);
}

\draw[gray, thin] (qi1) -- (qi2) -- (qij) -- (qi4) -- (qi5);
\draw[gray, thin] (qi1) to[bend angle=20, bend right] (qij) to[bend angle=20, bend right] (qi5) to[bend angle=30, bend left] (qi1); 
\draw[gray, thin] (qi2) to[bend angle=30, bend right] (qi4);

\foreach \x in {qj1,qj2,qji,qj4,qj5}
{
\draw[gray, thin] (vj) -- (\x);
}

\draw[gray, thin] (qj1) -- (qj2) -- (qji) -- (qj4) -- (qj5);
\draw[gray, thin] (qj1) to[bend angle=20, bend right] (qji) to[bend angle=20, bend right] (qj5) to[bend angle=30, bend left] (qj1); 
\draw[gray, thin] (qj2) to[bend angle=30, bend right] (qj4);

\draw[thick] (qij) to[bend angle=20, bend left] (qji);

\draw[dashed, thin, gray] (vi) -- (vj);
\draw[thin, gray] (4.5,9.4) rectangle (21.5,12.5);
\node at (7,11.5) {\textcolor{gray}{$G$}};
 
\end{tikzpicture}
\end{center}
\caption{The subgraph in $G'$ induced by $\{v_i\}\cup V(Q_i)\cup\{v_j\}\cup V(Q_j)$, $i\not= j$; $v_i$ and $v_j$ are adjacent in $G'$ if and only if they are adjacent in $G$.}\label{fig:ij}
\end{figure}
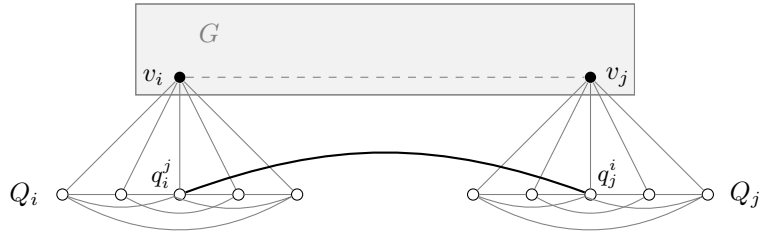

Observe that from the construction of $G'$ we have 
\begin{itemize}
\item[(p1)]
$N_{G'}(v_i)=N_G(v_i)\cup V(Q_i)$ for every vertex $v_i$ of $G$, and
\item[(p2)]
each vertex in $Q_i$ has at most one neighbor outside $V(Q_i)\cup\{v_i\}$. Namely, $N_{G'}(q_i^i)=(V(Q_i)\setminus\{q_i^i\})\cup\{v_i\}$, and for $j\not=i$,  $N_{G'}(q_i^j)=(V(Q_i)\setminus\{q_i^j\})\cup\{v_i, q_j^i\}$. 
\end{itemize}

\begin{lemma}\label{lem:G'diam3}
$G'$ has diameter $3$.
\end{lemma}
\begin{proof}
Let $x\not= y$ be two non-adjacent vertices of $G'$. Then $x\in V(Q_i)\cup\{v_i\}$ and $y\in V(Q_j)\cup\{v_j\}$ for some $i,j\in\{1,\ldots, n\}$, $i\not=j$. 
Then $G'[x,q^j_i, q^i_j,y]$ is a path in $G'$ of length at most three connecting $x$ and $y$. (See also Figure~\ref{fig:ij}.)  
Thus, $\d_{G'}(x,y)\le 3$ for all $x,y\in V'$. Note that, for $i\not= j$, $d_{G'}(q_i^i, q_j^j)=3$, hence $G'$ has diameter $3$. 
\end{proof}

\begin{lemma}\label{lem:G<->G'}
$G$ has a matching cut if and only if $G'$ has a matching cut.
\end{lemma}
\begin{proof}
Let $(X,Y)$ be a matching cut in $G$. Set 
\[
X':=X\cup \bigcup_{v_i\in X} Q_i\, \text{ and }\, Y':= Y\cup \bigcup_{v_i\in Y} Q_i.
\]
Then, by (p1) and (p2), $(X',Y')$ clearly is a matching cut in $G'$. 

Conversely, if $(X', Y')$ is a matching cut in $G'$, then for any $i$, the clique $V(Q_i)\cup\{v_i\}$ is contained in $X'$ or else in $Y'$. Hence $X=X'\cap V(G)$ and $Y=Y'\cap V(G)$ are both non-empty, and therefore $(X,Y)$ is a matching cut in $G$. 
\end{proof}

\medskip
By Lemmas~\ref{lem:G<->G'} and \ref{lem:G'diam3}, we conclude
\begin{lemma}\label{lem:diam3}
\MC\ is $\NP$-complete in graphs of diameter $3$.
\end{lemma}

Now, let $d\ge 4$ be a fixed integer, and let $B_d$ be a bipartite chain of $d-3$ complete bipartite graphs $K_{2,2}$ as depicted in Figure~\ref{fig:Bd}. 

\begin{figure}[ht]
\begin{center}
\begin{tikzpicture}[scale=.39]
\tikzstyle{vertex}=[draw,circle,inner sep=1.5pt]; 

\node[vertex] (r1)  at (0,3)  [label=above:$s_1$] {}; 
\node[vertex] (r1') at (0,0)  [label=below:$s_1'$] {}; 
\node[vertex] (r2)  at (3,3)  [label=above:$s_2$] {}; 
\node[vertex] (r2') at (3,0)  [label=below:$s_2'$] {}; 
\node[vertex] (r3)  at (6,3)  {}; 
\node[vertex] (r3') at (6,0)  {}; 
\node (x1)  at (6.2,3) {};
\node (x1') at (6.2,0) {};
\node (x2)  at (8.8,3) {};
\node (x2') at (8.8,0) {};
\node[vertex] (rd-4)  at (9,3) {}; 
\node[vertex] (rd-4') at (9,0) {}; 
\node[vertex] (rd-3)  at (12,3) [label=above:$s_{d-2}$] {}; 
\node[vertex] (rd-3') at (12,0) [label=below:$s_{d-2}'$] {}; 

\draw (r1) -- (r2) -- (r3);
\draw (r3') -- (r2') -- (r1');
\draw (r1) -- (r2') -- (r3);
\draw (r1') -- (r2) -- (r3'); 
\draw (rd-4) -- (rd-3) -- (rd-4') -- (rd-3') -- (rd-4);  

\draw[dashed,step=.1mm] (x1) -- (x2) -- (x1') -- (x2') -- (x1); 
\end{tikzpicture}
\end{center}
\caption{The bipartite chain $B_d$ of $d-3$ $K_{2,2}$s.}\label{fig:Bd}
\end{figure}
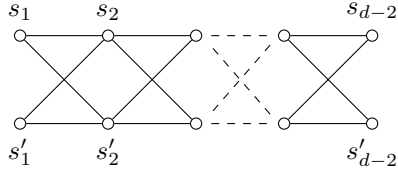

Let $G_d$ be obtained from $G'$ in Lemma~\ref{lem:G'diam3} and $B_d$ by joining the two vertices $s_1$ and $s_1'$ to all vertices in $Q_1$ of $G'$. 

Clearly, $\d_{G_d}(s_{d-2}, x)\le d$ for all $x\in V(G')$. Hence $\d_{G_d}(x,y)\le d$ for all vertices $x,y$ of $G_d$, and $\d_{G_d}(s_{d-2},v_2)=d$. Thus, $G_d$ has diameter $d$. Observe that each vertex of $B_d$ is contained, in $G_d$, in a $K_{2,3}$, hence in any matching cut $(X,Y)$ of $G_d$, $B_d+Q_1+v_1$ is contained in $G_d[X]$ or else in $G_d[Y]$. Thus, $G'$ has a matching cut if and only if $G_d$ has a matching cut. Hence, Lemma~\ref{lem:diam3} implies

\begin{theorem}\label{thm:diam3}
For any fixed $d\ge 3$, \MC\ is $\NP$-complete in graphs of diameter~$d$.
\end{theorem}


\section{Matching Cut in bipartite graphs of fixed diameter $d\ge 4$}\label{sec:diam4bip}

In this section, we first modify the restriction of \MC\ on diameter-$3$ graphs in the previous section to obtain a similar reduction to \MC\ on bipartite graphs of diameter $4$. 

Given a bipartite graph $G=(V_1,V_2,E)$ with two color classes $V_1$ and $V_2$, construct a new bipartite graph $G'=(V_1',V_2',E')$ as follows. 
First, let $V(G)=\{v_1,v_2,\ldots, v_n\}$, and, for each $1\le i\le n$, let $S_i=(Q_i, R_i, E_i)$ be a complete bipartite graph with color classes $Q_i=\{q^1_i, \ldots, q^n_i\}$ and $R_i=\{r^1_i, \ldots, r^n_i\}$. 
For $1\le i,j\le n, i\not=j$ write 
\[
E_{ij}=\begin{cases} 
          \{q^j_ir^i_j, r^j_iq^i_j\}, & \text{if $v_i, v_j\in V_1$ or $v_i, v_j\in V_2$},\\ \smallskip
          \{q^j_iq^i_j, r^j_ir^i_j\}, & \text{otherwise}.
        \end{cases}
\]

Then, $G'$ is obtained from $G$ and $S_i, 1\le i\le n$, by adding edges between $v_i$ and all vertices in $Q_i$ (thus, for each $i$, $S_i+v_i=(Q_i, R_i\cup\{v_i\})$ is a biclique in $G'$) and edge set $E_{ij}$ between $S_i$ and $S_j$ for any pair $i,j\in\{1,\ldots,n\}, i\not=j$. Formally,

\begin{align*}
  V' &=V\, \cup \bigcup_{1\le i\le n} Q_i\, \cup \bigcup_{1\le i\le n}  R_i,\\ \smallskip
  E' &=E\, \cup \bigcup_{1\le i\le n} \{v_iq^1_i,\ldots, v_iq^n_i\}\, \cup \bigcup_{1\le i\le n} E_i\, \cup \bigcup_{1\le i,j\le n, i\not=j} E_{ij}.
\end{align*}

See also Figures~\ref{fig:bip-ij} and~\ref{fig:bip-ij-different}. Clearly, $G'$ can be constructed from $G$ in $O(n^3)$ steps.

\begin{figure}[htb]
\begin{center}
\begin{tikzpicture}[scale=.39]
\tikzstyle{vertexG}=[circle,inner sep=1.5pt,fill=black];
\tikzstyle{vertexY}=[draw,circle,inner sep=1.5pt];
\tikzstyle{vertex}=[draw,circle,inner sep=1.5pt]; 

\filldraw[fill=black!5!white, draw=gray] (4.5,9.4) rectangle (21.5,12.5);

\node[vertexG] (vi) at (6,10)  [label=left:$v_i$] {}; 
\node at (0,6) {$Q_i$};
\draw[gray, thin] (6,6) ellipse (140pt and 18pt);
\node[vertex] (qi1) at (2,6)  {}; 
\node[vertex] (qi2) at (4,6)  {}; 
\node[vertex] (qij) at (6,6)  [label=left:\small $q^j_i$] {}; 
\node[vertex] (qi4) at (8,6)  {}; 
\node[vertex] (qi5) at (10,6)  {}; 
\node at (0,2) {$R_i$};
\draw[gray, thin] (6,2) ellipse (140pt and 18pt);
\node[vertex] (ri1) at (2,2)  {};
\node[vertex] (ri2) at (4,2)  {}; 
\node[vertex] (rij) at (6,2)  [label=left:\small $r^j_i$] {}; 
\node[vertex] (ri4) at (8,2)  {}; 
\node[vertex] (ri5) at (10,2)  {}; 

\node[vertexG] (vj) at (20,10)  [label=right:$v_j$] {}; 
\node at (26,6) {$Q_j$};
\draw[gray, thin] (20,6) ellipse (140pt and 18pt);
\node[vertex] (qj1) at (16,6)  {}; 
\node[vertex] (qj2) at (18,6)  {}; 
\node[vertex] (qji) at (20,6)  [label=right:\small $q^i_j$] {};
\node[vertex] (qj4) at (22,6)  {}; 
\node[vertex] (qj5) at (24,6)  {}; 
\node at (26,2) {$R_j$};
\draw[gray, thin] (20,2) ellipse (140pt and 18pt);
\node[vertex] (rj1) at (16,2)  {}; 
\node[vertex] (rj2) at (18,2)  {}; 
\node[vertex] (rji) at (20,2)  [label=right:\small $r^i_j$] {}; 
\node[vertex] (rj4) at (22,2)  {}; 
\node[vertex] (rj5) at (24,2)  {}; 

\foreach \x in {qi1,qi2,qij,qi4,qi5}
{
\draw[gray, thin] (vi) -- (\x);
}

\foreach \x in {qi1,qi2,qij,qi4,qi5}
{
\foreach \y in {ri1,ri2,rij,ri4,ri5}
{
\draw[gray, thin] (\x) -- (\y);
}
}

\foreach \x in {qj1,qj2,qji,qj4,qj5}
{
\draw[gray, thin] (vj) -- (\x);
}

\foreach \x in {qj1,qj2,qji,qj4,qj5}
{
\foreach \y in {rj1,rj2,rji,rj4,rj5}
{
\draw[gray, thin] (\x) -- (\y);
}
}

\draw[thick] (qij) -- (rji); \draw[thick] (rij) -- (qji);

\draw[thin, gray] (4.5,9.4) rectangle (21.5,12.5);
\node at (10,11.5) {\textcolor{gray}{$G=(V_1, V_2, E)$}};
 
\end{tikzpicture}
\end{center}
\caption{The subgraph in $G'$ induced by $\{v_i\}\cup S_i\cup\{v_j\}\cup S_j$, $i\not= j$, where $v_i$ and $v_j$ belong to the same color class of $G$.}\label{fig:bip-ij}
\end{figure}
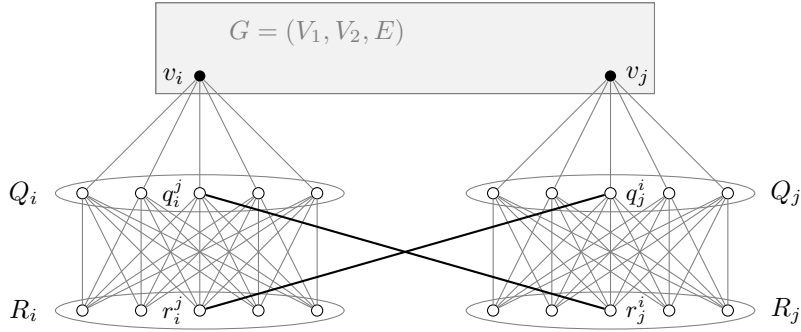

\begin{figure}[htb]
\begin{center}
\begin{tikzpicture}[scale=.39]
\tikzstyle{vertexG}=[circle,inner sep=1.5pt,fill=black];
\tikzstyle{vertexY}=[draw,circle,inner sep=1.5pt];
\tikzstyle{vertex}=[draw,circle,inner sep=1.5pt]; 

\filldraw[fill=black!5!white, draw=gray] (4.5,9.4) rectangle (21.5,12.5);

\node[vertexG] (vi) at (6,10)  [label=left:$v_i$] {}; 
\node at (0,6) {$Q_i$};
\draw[gray, thin] (6,6) ellipse (140pt and 18pt);
\node[vertex] (qi1) at (2,6)  {}; 
\node[vertex] (qi2) at (4,6)  {}; 
\node[vertex] (qij) at (6,6)  [label=left:\small $q^j_i$] {}; 
\node[vertex] (qi4) at (8,6)  {}; 
\node[vertex] (qi5) at (10,6)  {}; 
\node at (0,2) {$R_i$};
\draw[gray, thin] (6,2) ellipse (140pt and 18pt);
\node[vertex] (ri1) at (2,2)  {}; 
\node[vertex] (ri2) at (4,2)  {}; 
\node[vertex] (rij) at (6,2)  [label=left:\small $r^j_i$] {}; 
\node[vertex] (ri4) at (8,2)  {}; 
\node[vertex] (ri5) at (10,2)  {}; 

\node[vertexG] (vj) at (20,10)  [label=right:$v_j$] {}; 
\node at (26,6) {$Q_j$};
\draw[gray, thin] (20,6) ellipse (140pt and 18pt);
\node[vertex] (qj1) at (16,6)  {}; 
\node[vertex] (qj2) at (18,6)  {};
\node[vertex] (qji) at (20,6)  [label=right:\small $q^i_j$] {}; 
\node[vertex] (qj4) at (22,6)  {}; 
\node[vertex] (qj5) at (24,6)  {}; 
\node at (26,2) {$R_j$};
\draw[gray, thin] (20,2) ellipse (140pt and 18pt);
\node[vertex] (rj1) at (16,2)  {}; 
\node[vertex] (rj2) at (18,2)  {}; 
\node[vertex] (rji) at (20,2)  [label=right:\small $r^i_j$] {}; 
\node[vertex] (rj4) at (22,2)  {}; 
\node[vertex] (rj5) at (24,2)  {}; 

\foreach \x in {qi1,qi2,qij,qi4,qi5}
{
\draw[gray, thin] (vi) -- (\x);
}

\foreach \x in {qi1,qi2,qij,qi4,qi5}
{
\foreach \y in {ri1,ri2,rij,ri4,ri5}
{
\draw[gray, thin] (\x) -- (\y);
}
}

\foreach \x in {qj1,qj2,qji,qj4,qj5}
{
\draw[gray, thin] (vj) -- (\x);
}

\foreach \x in {qj1,qj2,qji,qj4,qj5}
{
\foreach \y in {rj1,rj2,rji,rj4,rj5}
{
\draw[gray, thin] (\x) -- (\y);
}
}

\draw[thick] (qij) to[bend angle=20, bend left] (qji); \draw[thick] (rij) to[bend angle=20, bend right] (rji);

\draw[dashed, thin, gray] (vi) -- (vj);
\draw[thin, gray] (4.5,9.4) rectangle (21.5,12.5);
\node at (10,11.5) {\textcolor{gray}{$G=(V_1, V_2, E)$}};
 
\end{tikzpicture}
\end{center}
\caption{The subgraph in $G'$ induced by $\{v_i\}\cup S_i\cup\{v_j\}\cup S_j$, $i\not= j$, where $v_i$ and $v_j$ belong to different color classes of $G$; $v_i$ and $v_j$ may be adjacent or not.}\label{fig:bip-ij-different}
\end{figure}
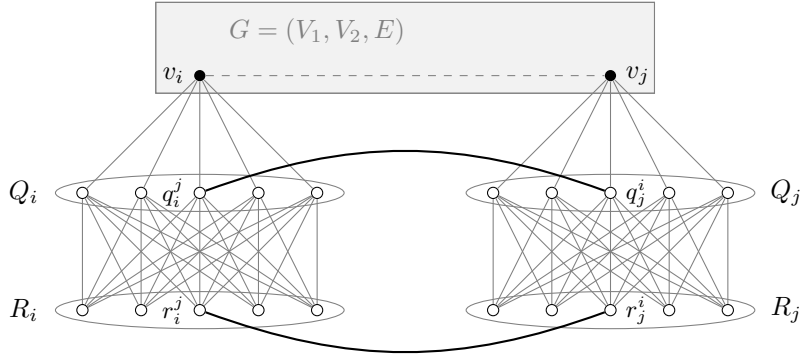

Observe that from the construction of $G'$ we have 
\begin{itemize}
\item[(p3)]
$N_{G'}(v_i)=N_G(v_i)\cup Q_i$ for every vertex $v_i$ of $G$, and
\item[(p4)]
each vertex in $S_i$ has at most one neighbor outside $V(S_i)\cup\{v_i\}$. Namely, 
$q_i^i$ and $r_i^i$ have no neighbor outside $V(S_i)\cup\{v_i\}$, and for $j\not=i$, $N_{G'}(q_i^j)=R_i\cup\{v_i, u\}$, where $u=r_j^i$ if $v_i$ and $v_j$ belong to the same color class of $G$, and $u=q_j^i$ otherwise.  $N_{G'}(r_i^j)=Q_i\cup\{w\}$, where $w=q_j^i$ if $v_i$ and $v_j$ belong to the same color class of $G$, and $w=r_j^i$ otherwise.   
\end{itemize}

\begin{lemma}\label{lem:bipG'diam4}
$G'$ is bipartite and has diameter $4$.
\end{lemma}
\begin{proof}
By construction,
\[
V_1':= V_1\,\cup \bigcup_{v_i\in V_1} R_i\,\cup \bigcup_{v_j\in V_2} Q_j\, \text{ and }\, V_2':= V_2\,\cup \bigcup_{v_i\in V_1} Q_i\,\cup \bigcup_{v_j\in V_2} R_j  
\]
are disjoint independent sets in $G'$ with $V'=V_1'\cup V_2'$, hence $G'$ is bipartite. 

We now show that $G'$ has diameter $4$. Let $x, y$ be two non-adjacent vertices of $G'$. Recall that, for every $1\le i\le n$, $C_i=(Q_i, R_i+v_i)$ is a biclique in $G'$. Hence we may assume that $x\in C_i$ and $y\in C_j$ for some $i\not=j$. (If $x,y\in C_i$ for some $i$, then they will have distance $2$.)

Suppose that $v_i$ and $v_j$ belong to the same color class of $G$. Then, if $x\in Q_i$ then $xr_i^jq_j^iP$, where $P$ is a path of length at most two in $C_j$ connecting $q_j^i$ and $y$, is a path of length at most $4$ connecting $x$ and $y$. If $x\in R_i\cup\{v_i\}$ then $xq_i^jr_j^iP'$, where $P'$ is a path of length at most two in $C_j$ connecting $r_j^i$ and $y$, is a path of length at most $4$ connecting $x$ and $y$. Similarly, it can be seen that in case $v_i$ and $v_j$ belong to different color classes of $G$ that there is a path of length at most $4$ connecting $x$ and $y$, too. 

Thus, $\d_{G'}(x,y)\le 4$ for all vertices $x$ and $y$ in $G'$. Note that for $i\not= j$ such that $v_i, v_j$ in the same color class of $G$, $d_{G'}(v_i, r)=4$ for any $r\in R_j\setminus\{r^i_j\}$, hence $G'$ has diameter~$4$. 
\end{proof}

\begin{lemma}\label{lem:bipG<->bipG'}
$G$ has a matching cut if and only if $G'$ has a matching cut.
\end{lemma}
\begin{proof}
Let $(X,Y)$ be a matching cut in $G$. Set 
\[
X':=X\,\cup \bigcup_{v_i\in X} V(S_i)\, \text{ and }\,  Y':= Y\,\cup \bigcup_{v_i\in Y} V(S_i).
\] 
Then, by (p3) and (p4), $(X',Y')$ clearly is a matching cut in $G'$. 

Conversely, if $(X', Y')$ is a matching cut in $G'$, then for any $i$, the biclique $(Q_i,R_i+v_i)$ is contained in $X'$ or else in $Y'$. 
Hence $X=X'\cap V(G)$ and $Y=Y'\cap V(G)$ are both non-empty, and therefore $(X,Y)$ is a matching cut in $G$. 
\end{proof}

\medskip
By Lemmas~\ref{lem:bipG<->bipG'} and \ref{lem:bipG'diam4}, we conclude

\begin{lemma}\label{lem:bipdiam4}
\MC\ is $\NP$-complete in bipartite graphs of diameter~$4$.
\end{lemma}

Now, for any fixed $d\ge 5$, let $G_d$ be obtained from $G'$ in Lemma~\ref{lem:bipG'diam4} and $B_{d-1}$ (see Section~\ref{sec:diam3}, Figure~\ref{fig:Bd}) by joining the two vertices $s_1$ and $s_1'$ to all vertices in $R_1$ of $G'$. 

Clearly, $G_d$ is bipartite. Moreover, $\d_{G_d}(s_{d-3}, x)\le d$ for all $x\in V(G')$. Hence $\d_{G_d}(x,y)\le d$ for all vertices $x,y$ of $G_d$, and $\d_{G_d}(s_{d-3},v_i)=d$ for all vertices $v_i$ not in the color class of $v_1$. Thus, $G_d$ has diameter $d$. Again, observe that each vertex of $B_{d-1}$ is contained, in $G_d$, in a $K_{2,3}$, hence in any matching cut $(X,Y)$ of $G_d$, $B_{d-1}+S_1+v_1$ is contained in $G_d[X]$ or else in $G_d[Y]$. Thus, $G'$ has a matching cut if and only if $G_d$ has a matching cut. Hence, Lemma~\ref{lem:bipdiam4} implies

\begin{theorem}\label{thm:bipdiam4}
For any fixed $d\ge 4$, \MC\ is $\NP$-complete in bipartite graphs of diameter~$d$.
\end{theorem}


\section{Matching Cut in bipartite graphs of diameter at most 3}\label{sec:diam3bip}


In this section we prove that \MC\ can be solved in polynomial time when restricted to bipartite graphs of diameter at most $3$. Our strategy is to decide in polynomial time, given two disjoint vertex subsets $A$ and $B$ of the input graph $G$, if $G$ admits a matching cut such that $A$ is contained in one part and $B$ is in the other part of the matching cut. 
To do this, we first prove a lemma (Lemma~\ref{lem:useful}) that will be useful in many cases. This useful lemma roughly says that, if after certain forcing rules are no longer applicable the connected components of the remaining part of the graph induced by the remaining unforced vertices are \lq monochromatic\rq, we will be able to reduce the problem to 2-SAT.  
As a by-product, we will derive from this lemma a new and simple polynomial-time algorithm solving \MC\ in graphs of diameter 2.

\subsection{A useful lemma}
Given a connected graph $G=(V,E)$ and two disjoint, non-empty vertex sets $A, B\subset V$ such that each vertex in $A$ is adjacent to exactly one vertex of $B$ and each vertex in $B$ is adjacent to exactly one vertex of $A$. We say a matching cut of $G$ is an \emph{$A$-$B$-matching cut} if $A$ is contained in one side and $B$ is contained in the other side of the matching cut. 
Observe that if $G$ has a matching cut, then $G$ has an $A$-$B$-matching with, for example, $A=\{a\}$ and $B=\{b\}$ for some edge $ab$ of $G$. Thus, in general, unless NP $\not=$ P, we cannot decide in polynomial time if $G$ admits an $A$-$B$-matching cut for a given pair $A, B$. However, there are some rules that force certain vertices some of which together with $A$ must belong to one side and the other together with $B$ must belong to the other side of such an $A$-$B$-matching cut (if any). We are going to describe such forcing rules.

Now assume that $A$ and $B$ are two disjoint, non-empty subsets of $V(G)$ such that each vertex in $A$ is adjacent to exactly one vertex of $B$ and each vertex in $B$ is adjacent to exactly one vertex of $A$. 
Initially, set $X:=A$, $Y:=B$ and write $R=V(G)\setminus (X\cup Y)$. The sets $A,B,X,Y$ will be extended, if possible, by adding vertices from $R$ according to the forcing rules such that the following properties hold before and after each update:
\begin{itemize}
\item[(f1)] $A\subseteq X$, $B\subseteq Y$ and $X\cap Y\not=\emptyset$;
\item[(f2)] each vertex in $A$ is adjacent to exactly one vertex of $B$ and each vertex in $B$ is adjacent to exactly one vertex of $A$;
\item[(f3)] no vertex in $X\setminus A$ is adjacent to a vertex in $Y$ and no vertex in $Y\setminus B$ is adjacent to a vertex in $X$;
\item[(f4)] any $A$-$B$-matching cut of $G$ must contain $X$ in one side and $Y$ in the other side.
\end{itemize}

These properties imply that after each update, $(X,Y)$ is an $A$-$B$-matching cut of $G[X\cup Y]$. The first three rules will detect certain vertices that ensure that $G$ cannot have an $A$-$B$-matching cut.  
\begin{description}
\item[(R1)] Let $v\in R$ be adjacent to a vertex in $A$. If $v$ is
   \begin{itemize}
      \item adjacent to a vertex in $B$, or 
      \item adjacent to (at least) two vertices in $Y\setminus B$,
   \end{itemize}
   then $G$ has no $A$-$B$-matching cut.
   
\item[(R2)] Let $v\in R$ be adjacent to a vertex in $B$. If $v$ is
   \begin{itemize}
      \item adjacent to a vertex in $A$, or 
      \item adjacent to (at least) two vertices in $X\setminus A$,
   \end{itemize}
   then $G$ has no $A$-$B$-matching cut.

\item[(R3)] If $v\in R$ is adjacent to (at least) two vertices in $X\setminus A$ and to (at least) two vertices in $Y\setminus B$, then $G$ has no $A$-$B$-matching cut.
\end{description}

The correctness of (R1), (R2) and (R3) is quite obvious. We assume that, before each application of the forcing rules (R4) and (R5) below, none of (R1), (R2) and (R3) is applicable.

\begin{description}
\item[(R4)] Let $v\in R$ be adjacent to a vertex in $A$ or to (at least) two vertices in $X\setminus A$. Then $X:=X\cup\{v\}$, $R:=R\setminus\{v\}$. If, moreover, $v$ has a unique neighbor $w\in Y\setminus B$ then $A:=A\cup\{v\}$, $B:=B\cup\{w\}$.

\item[(R5)] Let $v\in R$ be adjacent to a vertex in $B$ or to (at least) two vertices in $Y\setminus B$. Then $Y:=Y\cup\{v\}$, $R:=R\setminus\{v\}$. If, moreover, $v$ has a unique neighbor $w\in X\setminus A$ then $B:=B\cup\{v\}$, $A:=A\cup\{w\}$. 
\end{description}

To see that (R4) is correct, let, by induction hypothesis, $A$, $B$, $X$ and $Y$ fulfill the properties (f1) -- (f4). 
We have to argue that after a next application of (R4), $A$, $B$, $X$ and $Y$ still fulfill (f1) -- (f4). 
Consider an arbitrary $A$-$B$-matching cut $(X',Y')$ of $G$, and let, by (f4), $X\subseteq X'$ and $Y\subseteq Y'$, say. In particular, by (f1), $A\subseteq X'$ and $B\subseteq Y'$.  
Let $v\in R$ be adjacent to a vertex $a\in A$ or to (at least) two vertices in $X\setminus A$. (Then, since none of (R1), (R2), (R3) is applicable, $N(v)\cap B=\emptyset$ and $|N(v)\cap Y|\le 1$.) 
If $v$ is adjacent to a vertex $a\in A$, then, by (f2), $a$ has a neighbor in $B\subseteq Y'$. Hence $v$  must belong to $X'$ (as $(X',Y')$ is a matching cut of $G$), and (R4) correctly extends $X$ to $X\cup\{v\}\subseteq X'$. 
If $v$ is adjacent to two vertices in $X\setminus A\subset X'$, then clearly $v$ must belong to $X'$ (as $(X',Y')$ is a matching cut of $G$), and (R4) again correctly extends $X$ to $X\cup\{v\}\subseteq X'$. 
Moreover, in case $v$ has a unique neighbor $w\in Y\setminus B\subset Y'$, (R4) correctly extends $A$ to $A\cup\{v\}\subseteq X'$ and $B$ to $B\cup\{w\}\subseteq Y'$. Hence after updating $A$, $B$, $X$ and $Y$, $(X', Y')$ is still an $A$-$B$-matching cut of $G$ with $X\subseteq X', Y\subseteq Y'$. Thus, properties (f1)-- (f4) hold after an application of (R4). The correctness of (R5) can be seen similarly.

If none of (R1), (R2) and (R3) is applicable, then each vertex $v\in R$ has no neighbor in $A$ or has no neighbor in $B$, and $v$ has at most one neighbor in $X\setminus A$ or has at most one neighbor in $Y\setminus B$. If (R4) is not applicable, then each vertex $v\in R$ has no neighbor in $A$ and at most one neighbor in $X\setminus A$. If (R5) is not applicable, then each vertex $v\in R$ has no neighbor in $B$ and at most one neighbor in $Y\setminus B$. Thus, together with (f1) -- (f4), the following fact holds:

\begin{fact}\label{fact}
Suppose none of (R1) -- (R5) is applicable. Then 
\begin{itemize}
  \item $(X,Y)$ is an $A$-$B$-matching cut of $G[X\cup Y]$, and any $A$-$B$-matching cut of $G$ must contain $X$ in one side and $Y$ in the other side;
  \item for any vertex $v\in R$, $N(v)\cap A=N(v)\cap B=\emptyset$ and $|N(v)\cap X|\le 1,\, |N(v)\cap Y|\le 1$.
\end{itemize}
\end{fact}

We now bound the running time for the application of the rules. Observe first that the applicability of the rules can be tested in time $\sum_{v\in R} O(\deg(v))=O(|E|)$.  Each of the rules (R1), (R2) and (R3) can be applied in constant time and applies at most once. 
Each of the rules (R4) and (R5) can be applied in time $O(\deg(v))$ and applies at most $|V|$ many times because it removes one vertex from $R$. So, the total running time for applying the rules (R1) -- (R5) until they are no longer applicable is bounded by $|V|\cdot\sum_{v\in V} O(\deg(v))=O(|V|\!\cdot\!|E|)$. 

We say that a subset $S\subseteq R$ is \emph{monochromatic} if, in any $A$-$B$-matching cut of $G$, all vertices of $S$ belong to the same side of this matching cut.

\begin{lemma}\label{lem:useful}
Suppose none of the forcing rules (R1) -- (R5) is applicable. Assuming each connected component of $G-(X\cup Y)$ is monochromatic, it can be decided in time $O(|V|\!\cdot\! |E|)$ if $G$ admits an $A$-$B$-matching cut. 
\end{lemma}
\begin{proof} 
Let $Z$ be a connected component of $G-(X\cup Y)$, and let $(X',Y')$ be an $A$-$B$-matching cut of $G$ with $A\subseteq X'$ and $B\subseteq Y'$. Note, by Fact~\ref{fact}, that $X\subseteq X'$ and $Y\subseteq Y'$. Since $Z$ is monochromatic, we have 

\begin{itemize}
\item $Z\subseteq X'$ whenever some vertex in $X\setminus A$ has at least two neighbors in $Z$. 
\item Similarly, $Z\subseteq Y'$ whenever some vertex in $Y\setminus B$ has at least two neighbors in $Z$. 
\item If a vertex in $X\setminus A$ has neighbors in two connected components of $G-(X\cup Y)$, then at least one of these components is contained in $X'$. 
\item Similarly, if a vertex in $Y\setminus B$ has neighbors in two connected components of $G-(X\cup Y)$, then at least one of these components is contained in $Y'$. 
\end{itemize}

Thus, we can decide if $G$ admits a matching cut $(X',Y')$ such that $X\subseteq X', Y\subseteq Y'$, by solving the following instance $F(G)$ of the 2-SAT problem.

\begin{itemize}
\item For each connected component $C$ of $G-(X\cup Y)$, create two Boolean variables $x_C$ and $y_C$. The intention is that $x_C$ is set to $\mathrm{true}$ if $C$ must go to $X$ and $y_C$ is set to $\mathrm{true}$ if $C$ must go to $Y$. Then $(x_C\lor y_C)$ and $(\neg x_C\lor \neg y_C)$ are two clauses of the formula $F(G)$.
\item For each connected component $C$ of $G-(X\cup Y)$ with $|N(v)\cap C|\ge 2$ for some $v\in X\setminus A$, $(x_C)$ is a clause of the formula $F(G)$. This clause ensures that in this case, $C$ must go to $X$.
\item For each connected component $C$ of $G-(X\cup Y)$ with $|N(w)\cap C|\ge 2$ for some $w\in Y\setminus B$, $(y_C)$ is a clause of the formula $F(G)$. This clause ensures that in this case, $C$ must go to $Y$.
\item For each two connected components $C\not=D$ of $G-(X\cup Y)$ having a common neighbor in $X\setminus A$, $(x_C\lor x_D)$ is a clause of the formula $F(G)$. This clause ensures that in this case, at least one of $C$ and $D$ must go to $X$.
\item For each two connected components $C\not=D$ of $G-(X\cup Y)$ having a common neighbor in $Y\setminus B$, $(y_C\lor y_D)$ is a clause of the formula $F(G)$. This clause ensures that in this case, at least one of $C$ and $D$ must go to $Y$.
\end{itemize}

\begin{claim} 
$G$ admits a matching cut $(X',Y')$ such that $X\subseteq X', Y\subseteq Y'$ if and only if the $2$-CNF formula $F(G)$ is satisfiable. 
\end{claim}
{\em Proof of the Claim:}\, 
First, suppose $(X',Y')$ is a matching cut of $G$ with $X\subseteq X', Y\subseteq Y'$. 
For each connected component $C$ of $G-(X\cup Y)$, set $\mathrm{b}(x_C)=\mathrm{true}$ and $\mathrm{b}(y_C)=\mathrm{false}$ if $C\subset X'$, and $\mathrm{b}(x_C)=\mathrm{false}$ and $\mathrm{b}(y_C)=\mathrm{true}$ if $C\subset Y'$. Since each connected component $C$ is contained in $X'$ or else in $Y'$, the assignment $\mathrm{b}$ is well-defined. Moreover, $\mathrm{b}$ is a satisfying truth assignment for $F(G)$: 

\begin{itemize}
\item For each connected component $C$, the two corresponding clauses $(x_C\lor y_C)$ and $(\neg x_C\lor \neg y_C)$ are obviously satisfied by~$\mathrm{b}$.
\item For each connected component $C$ such that $|N(v)\cap C|\ge 2$ for some $v\in X$, $C\subset X'$ because $(X',Y')$ is a matching cut and $C$ is monochromatic. Hence the corresponding clause $(x_C)$ is satisfied by~$\mathrm{b}$.
\item Similarly, for each connected component $C$ such that $|N(w)\cap C|\ge 2$ for some $w\in Y$, $C\subset Y'$. Hence the corresponding clause $(y_C)$ is satisfied by~$\mathrm{b}$.
\item For each two connected components $C\not=D$ having a common neighbor in $X$, $C\subset X'$ or $D\subset X'$ because $(X',Y')$ is a matching cut and $C$ and $D$ are monochromatic. Hence the corresponding clause $(x_C\lor x_D)$ is satisfied by $\mathrm{b}$.
\item Similarly, for each two connected components $C\not=D$ having a common neighbor in $Y$, $C\subset Y'$ or $D\subset Y'$. Hence the corresponding clause $(y_C\lor y_D)$ is satisfied by~$\mathrm{b}$.
\end{itemize}

Thus, $F(G)$ is satisfied by $\mathrm{b}$.

Second, suppose $F(G)$ is satisfied, and let $\mathrm{b}$ be a satisfying truth assignment for $F(G)$. Then let $X'$ consist of $X$ and all connected components $C$ of $G-(X\cup Y)$ with $\mathrm{b}(x_C)=\mathrm{true}$ and let $Y'$ consist of $Y$ and all connected components $C$ of $G-(X\cup Y)$ with $\mathrm{b}(y_C)=\mathrm{true}$. Since for each connected component $C$ of $G-(X\cup Y)$, the two corresponding clauses $(x_C\lor y_C)$ and $(\neg x_C\lor \neg y_C)$ are satisfied by $\mathrm{b}$, $(X',Y')$ is a partition of $V(G)$. 
Moreover $(X',Y')$ is a matching cut of $G$: Assume, for a contradiction, some $v\in X'$ has two neighbors $u_1, u_2$ in $Y'$. 

If $v\in X$, then, by Fact~\ref{fact}, $v\in X\setminus A$ and $u_1,u_2\in Y'\setminus Y$. Let $C_1$ and $C_2$ be the connected components containing $u_1$ and $u_2$, respectively. If $C_1=C_2=:C$, then $(x_C)$ is a clause of $F(G)$ and therefore $\mathrm{b}(x_C)=\mathrm{true}$, hence $C\subset X'$, contradicting $u_1,u_2\in Y'$. If $C_1\not=C_2$, then $(x_{C_1}\lor x_{C_2})$ is a clause of $F(G)$, hence $\mathrm{b}(x_{C_1})=\mathrm{true}$ or $\mathrm{b}(x_{C_2})=\mathrm{true}$, contradicting $u_1,u_2\in Y'$ again. 

If $v\in X'\setminus X$, then let $C$ be the connected component containing $v$. Then, by Fact~\ref{fact} again, at most one of $u_1, u_2$ is in $Y$, say $u_1\in Y'\setminus Y$. Thus, $u_1\in C$, hence $C\not\subset X'$ and $C\not\subset Y'$. This contradicts the definition of $(X',Y')$.

Thus, each vertex in $X'$ has at most one neighbor in $Y'$, and, similarly, each vertex in $Y'$ has at most one neighbor in $X'$. That is, $(X',Y')$ is a matching cut of $G$ with $X\subseteq X', Y\subseteq Y'$. 

The proof of the Claim is complete.

Note that $F(G)$ has $O(|V|)$ variables and $O(|V|^2)$ clauses and 
can be constructed in time $O(|V|(|V|+|E|))=O(|V|\!\cdot\!|E|)$. 
Since 2-SAT can be solved in linear time (cf.~\cite{APT,DavPut,EvenIS76}), by the Claim above 
we can decide in time $O(|V|\!\cdot\!|E|)$ if $G$ admits an $A$-$B$-matching cut. 

Lemma~\ref{lem:useful} is proved. 
\end{proof}

\subsection{Diameter 2 graphs: A new, simple and faster polynomial-time algorithm}

Let $G=(V,E)$ be a graph of diameter 2. Choose an edge $ab$ of $G$, and apply rules (R1) -- (R5) for $A:=\{a\}$, $B:=\{b\}$ as long as possible. If (R1) or (R2) or (R3) is applicable, then clearly $G$ has no $A$-$B$-matching cut. So let us assume that none of (R1), (R2) and (R3) was ever applied and none of (R4) and (R5) is applicable. Then each connected component $Z$ of $G-(X\cup Y)$ is monochromatic. To see this, let $(X',Y')$ be an $A$-$B$-matching cut of $G$ with $X\subset X'$, $Y\subset Y'$. By Fact~\ref{fact}, any vertex in $A\cup B$ is non-adjacent to any vertex in $Z$, any vertex in $A$ has neighbors only in $X\cup B$, any vertex in $B$ has neighbors only in $Y\cup A$. Since $G$ has diameter 2, therefore, $N(v)\cap N(a)\not=\emptyset$ for all $v\in Z$ and $a\in A$. Thus, $v$ must have a neighbor  in $X\setminus A$. Similarly, each vertex $v\in Z$ must have a neighbor in $Y\setminus B$. Now, suppose to the contrary that $Z$ is not monochromatic. Then, by connectedness, there is an edge $uv$ in $Z$ with $u\in X'$ and $v\in Y'$, say. But then $u, v$ and a neighbor of $v$ in $X\subseteq X'$ induce a bad $P_3$, contradicting the assumption that $(X',Y')$ is a matching cut. 

Thus, any connected component of $G-(X\cup Y)$ is monochromatic. Therefore, by Lemma~\ref{lem:useful}, it can be decided in time $O(|V|\!\cdot\!|E|)$ if $G$ has an $A$-$B$-matching cut. Since we have at most $|E|$ many choices for the edge $ab$, we conclude that \MC\ can be solved in time $O(|V|\!\cdot\!|E|^2)$ for graphs of diameter two. 
We remark that the known algorithm posed in~\cite{BorowieckiJ08} has slower running time $O(|V|^2|E|^3)$. Moreover, in comparison to their algorithm, ours is much simpler.

\subsection{Diameter 3 bipartite graphs} 

We are now going to describe how to solve \MC\ in polynomial time when restricted to bipartite graphs of diameter at most $3$. We will use the following characterization of diameter-3 bipartite graphs.
\begin{fact}\label{fact:diam3bip} 
Let $G=(V_1,V_2,E)$ be a connected bipartite graph with color classes $V_1$ and $V_2$. Then  $G$ has diameter at most $3$ if and only if, for each $i=1,2$ and for every two vertices $u,v\in V_i$, $N(u)\cap N(v)\not=\emptyset$.
\end{fact} 
\begin{proof}
First, let $G=(V_1,V_2,E)$ have diameter at most $3$. Consider two vertices $u, v\in V_1$. Then, as any $u,v$-path in $G$ has even length, there must be an $u,v$-path of length $2$, i.e., $N(u)\cap N(v)\not=\emptyset$. Conversely, suppose that $N(u)\cap N(v)\not=\emptyset$ for every two vertices $u,v\in V_i$. Consider two non-adjacent vertices $x, y$ with $x\in V_1$ and $y\in V_2$. Then, as $G$ is connected, $x$ has a neighbor $y'\in V_2-y$. Hence, since $N(y)\cap N(y')\not=\emptyset$, there is an $x,y$-path of length $3$. Thus, there is a path of length at most $3$ between every two vertices of $G$, i.e., $G$ has diameter at most $3$.  
\end{proof}

\smallskip
Let $G=(V_1,V_2,E)$ have diameter at most $3$. Since graphs having a bridge have a matching cut, we may assume that $G$ is 2-edge connected. Hence every matching cut of $G$, if any, must have at least two edges. Our  algorithm consists of two phases. 
In the phase 1, we will check if $G$ has a matching cut containing two edges $a_1b_1, a_2b_2$ such that $a_1,a_2\in V_1$, $b_1,b_2\in V_2$ and $\{a_1,b_2\}$ is in one part and $\{a_2,b_1\}$ is in the other part on the matching cut. To do this, we start with $A=\{a_1,b_2\}$ and $B=\{a_2,b_1\}$ and apply the forcing rules as long as we can. Then we will see that, if $G$ has an $A$-$B$-matching cut, the condition of Lemma~\ref{lem:useful} will be fulfilled and so we can reduce the problem to 2-SAT.  
In case phase 1 is unsuccessful, that is, $G$ has no $A$-$B$-matching cut such that both $A$ and $B$ contain vertices of both color classes $V_1$ and $V_2$, phase 2 will be started. In the phase 2, we will check if $G$ has an $A$-$B$-matching cut such that $A\subseteq V_1$ and $B\subseteq V_2$ (or $A\subseteq V_2$ and $B\subseteq V_1$). In this phase, we will see that such an $A$-$B$-matching cut of $G$, if any, can easily be found because the partition of unforced vertices is uniquely determined by the color classes of $G$. The details as follows.

\paragraph{Phase 1}  
Choose two edges $a_1b_1, a_2b_2\in E$ such that $a_1, a_2\in V_1$ and $b_1,b_2\in V_2$. Set $X=A=\{a_1, b_2\}$, $Y=B=\{a_2,b_1\}$ and write $R=V(G)\setminus (X\cup Y)$. Apply (R1) -- (R5) as long as possible, and let us assume that none of (R1), (R2) and (R3) was ever applied. 

Then each connected component $Z$ of $G-(X\cup Y)$ is monochromatic. To see this, consider an arbitrary vertex $v\in Z$. If $v\in V_1$, then, since $a_1\in A\cap V_1$, $N(v)\cap N(a_1)\not=\emptyset$ by Fact~\ref{fact:diam3bip}. By Fact~\ref{fact}, $N(v)\cap (A\cup B)=\emptyset$ and $N(a_1)\subseteq X\cup B$, hence $N(v)\cap N(a_1)\subseteq X\setminus A$, i.e., $v$ has a neighbor in $X\setminus A$. Similarly, since $a_2\in B\cap V_1$, $v$ has a neighbor in $Y\setminus B$. Similarly, if $v\in V_2$, then, since $b_1\in A\cap V_2$ and $b_2\in B\cap V_2$, $v$ has a neighbor in $X\setminus A$ and a neighbor in $Y\setminus B$, too. Thus, every vertex in $Z$ has a neighbor in $X$ and a neighbor in $Y$. Therefore, by the same argument explained in the diameter 2 case, $Z$ is monochromatic: let $(X',Y')$ be an $A$-$B$-matching cut of $G$ with, by Fact~\ref{fact}, $X\subseteq X'$ and $Y\subseteq Y'$, say, and assume to the contrary that $Z\cap X'\not=\emptyset$ and $Z\cap Y'\not=\emptyset$. Then, by connectedness, there is an edge $uv$ in $Z$ with $u\in X'$, $v\in Y'$, say. But then $u, v$ and a neighbor of $v$ in $X\subseteq X'$ induce a bad $P_3$, contradicting the assumption that $(X',Y')$ is a matching cut. 

Hence, by Lemma~\ref{lem:useful}, we can decide in time $O(|V|\!\cdot\!|E|)$ if $G$ has an $A$-$B$-matching cut. Since there are at most $|E|^2$ choices for $A$ and $B$, we conclude that, with phase 1, we can decide in time $O(|V|\!\cdot\!|E|^3)$ if $G$ has a matching cut $(X',Y')$ containing two edges $a_1b_1, a_2b_2$ such that $a_1,a_2\in V_1$, $b_1,b_2\in V_2$ and $\{a_1,b_2\}\subseteq X'$ and $\{a_2,b_1\}\subseteq Y'$. 

\paragraph{Phase 2}  
In this second phase we assume that phase 1 is unsuccessful, that is, $G$ has no matching cut $(X',Y')$ containing two edges $a_1b_1, a_2b_2$ such that $a_1,a_2\in V_1$, $b_1,b_2\in V_2$ and $\{a_1,b_2\}\subseteq X'$ and $\{a_2,b_1\}\subseteq Y'$.

Choose an edge $ab\in E$ with $a\in V_1$ and $b\in V_2$. Set $X=A=\{a\}$, $Y=B=\{b\}$ and write $R=V(G)\setminus (X\cup Y)$. Apply (R1) -- (R5) as long as possible, and let us assume that none of (R1), (R2) and (R3) was ever applied. (Note that, as $G$ is $2$-edge connected, $A$ and $B$ will be updated at least once by forcing rule (R4) or (R5), assuming $G$ has an $A$-$B$-matching cut.)

Then, by Facts~\ref{fact} and \ref{fact:diam3bip} (cf. also phase 1), every vertex $v\in R_1=R\cap V_1$ has a neighbor in $X$ (as $a\in A\cap V_1$) and every vertex $w\in R_2=R\cap V_2$ has a neighbor in $Y$ (as $b\in B\cap V_2$). Thus, assuming $G$ has an $A$-$B$-matching cut $(X',Y')$ with $X\subseteq X'$ and $Y\subseteq Y'$, $R_1$ must be included in $X'$ and $R_2$ must be included in $Y'$. For, if $v\in R_1$ was in $Y'$, then the $A$-$B$-matching cut $(X',Y')$ would contain the edges $ab$ and $uv$, where $u$ is the neighbor of $v$ in $X\subseteq X'$, with $a,v\in V_1$ and $b,u\in V_2$, contradicting the assumption that phase 1 was unsuccessful. The case of $R_2$ is completely similar. 
Therefore, $G$ has an $A$-$B$-matching cut $(X',Y')$ with $X\subseteq X', Y\subseteq Y'$ if and only if $(X\cup R_1, Y\cup R_2)$ is a matching cut. As the second property can be checked in linear time, and there are at most $|E|$ choices for $A$ and $B$, we conclude that phase 2 can be performed in time $O(|V|\!\cdot\!|E|^2)$ for deciding if $G$ has a matching cut. 

Putting all together we obtain:
\begin{theorem}\label{thm:diam3bip}
\MC\ can be solved in time $O(|V|\!\cdot\!|E|^3)$ when restricted to bipartite graphs of diameter at most~$3$.
\end{theorem}


\section{Concluding remarks}\label{sec:con}


In this paper we have completely determined the computational complexity of \MC\ for graphs and bipartite graphs with respect to diameter constraint: \MC\ is $\NP$-complete when restricted to graphs of diameter $d$, for fixed $d\ge 3$, and to bipartite graphs of diameter $d$, for fixed diameter $d\ge 4$. In the other cases, \MC\ can be solved in polynomial time. 

We remark that \MC\ can be solved in linear time in planar graphs (and in graphs with fixed genus) of fixed diameter. This is because a planar graph with diameter $d$ has tree-width at most $3d-2$ (\cite{Eppstein99}). (More generally, a graph with diameter $d$ and genus $g$ has tree-width $O(gd)$; see~\cite{Eppstein00}.) In~\cite{Bonsma09}, it is shown that \MC\ can be expressed in monadic second order logic (MSOL), and it is well-known (\cite{ALS}) that all graph properties definable in MSOL can be decided in linear time for classes of graphs with bounded tree-width, when a tree-decomposition is given. It is also well-known (\cite{Bodlaender}) that a tree-decomposition of bounded width of a given graph can be found in linear time. Combining these facts, it follows that \MC\ can be solved in linear time for planar graphs (and in graphs with fixed genus) of fixed diameter.

Finally, we note that, in contrast, the computational complexity of \MC\ for graphs of given girth is still incompletely determined. By results due to Bonsma (\cite{Bonsma09}), 
\MC\ is $\NP$-complete for planar graphs of girth at most five and polynomially solvable for \emph{planar} graphs of girth at least six. (In fact, Bonsma showed that all planar graphs of girth $\ge 6$ have a matching cut.) In particular, \MC\ is $\NP$-complete for graphs of girth $\le 5$.

\smallskip
\noindent  
\textbf{Open Question:} What is the computational complexity of \MC\ for 
 graphs of girth at least six?


\begin{thebibliography}{}
\bibitem{ACGH}
J\'ulio C. Ara\'ujo, Nathann Cohen, Fr\'ed\'eric Giroire, Fr\'ed\'eric Havet, 
Good edge-labelling of graphs, Discrete Applied Mathematics 160 (2012) 2502--2513.

\bibitem{ALS}
Stefan Arnborg, Jens Lagergren, and Detlef Seese, Easy problems for tree-decomposable graphs, J. Algorithms 12 (1991) 308--340.

\bibitem{APT}
Bengt Aspvall, Michael F. Plass, Robert E. Tarjan, A linear-time algorithm for testing the truth of certain quantified boolean formulas, Information Processing Letters 8 (1979) 121-123. Erratum: 14 (1982) 195.

\bibitem{Bodlaender} Hans~L. Bodlaender, A linear-time algorithm for finding tree-decompositions
of small treewidth, SIAM J. Computing 25 (1996) 1305--1317.

\bibitem{Bonsma09} Paul Bonsma, The complexity of the Matching-Cut problem for planar graphs and other graph classes, J. Graph Theory 62 (2009) 109--126.

\bibitem{BorowieckiJ08} Mieczys{\l}aw Borowiecki, Katarzyna Jesse-J\'ozefczyk, Matching cutsets in graphs of diameter 2, Theoretical Computer Science 407 (2008) 574--582.

\bibitem{Chvatal84} Va\v{s}ek Chv\'atal, Recognizing decomposable graphs, J. Graph Theory 8 (1984) 51--53.

\bibitem{DavPut} Martin Davis, Hilary Putnam, A computing procedure for quantification theory, J. ACM 7 (1960) 201--215.

\bibitem{Eppstein99}
David Eppstein, Subgraph isomorphism in planar graphs and related problems. J. Graph Algorithms
Appl. 3 (1999) 1--27.

\bibitem{Eppstein00}
David Eppstein, Diameter and treewidth in minor-closed families, Algorithmica 27 (2000) 275--291.

\bibitem{EvenIS76} Shimon Even, Alon Itai, Adi Shamir, On the complexity of timetable and multicommodity flow problems, 
SIAM J. Computing 5 (1976) 691--703.

\bibitem{FarleyP82} Arthur M. Farley, Andrzej Proskurowski, Networks immune to isolated line failures, Networks 12 (1982) 393--403.

\bibitem{Graham} Ron L. Graham, On primitive graphs and optimal vertex assignments.
Ann. N.Y. Acad. Sci. 175 (1970) 170--186.

\bibitem{KratschL16} Dieter Kratsch, Van Bang Le, Algorithms solving the Matching Cut problem, Theoretical Computer Science 609 (2016) 328--335. 

\bibitem{LeL16}
Hoang-Oanh Le, Van Bang Le, On the complexity of Matching Cut in graphs of fixed diameter, In Proceedings of the 27th International Symposium on Algorithms and Computation, {ISAAC} 2016, December 12-14, 2016, Sydney, Australia, pp.  50:1--50:12, doi: 10.4230/LIPIcs.ISAAC.2016.50.

\bibitem{LeR03} Van Bang Le, Bert Randerath, On stable cutsets in line graphs, Theoretical Computer Science 301 (2003) 463--475.

\bibitem{Moshi89} Augustine~M. Moshi, Matching cutsets in graphs, J. Graph Theory 13 (1989) 527--536.

\bibitem{PatrignaniP01} Maurizio Patrignani, Maurizio Pizzonia, The complexity of the matching-cut problem, WG~2001, Boltenhagen, Lecture Notes in Computer Science 2204, Springer, Berlin, 2001, pp. 284--295.



\end{thebibliography}
\end{document}